\documentclass[conference]{IEEEtran}

\IEEEoverridecommandlockouts
\usepackage{amsmath}
\usepackage{enumitem}
\usepackage{blindtext}
\usepackage{amsfonts}
\usepackage{amssymb}
\usepackage{graphicx}
\usepackage{epstopdf}%
\usepackage[nice]{nicefrac}
\usepackage{svg}
\usepackage{array}
\usepackage{authblk}
\usepackage{tikz}
\usepackage[utf8]{inputenc}
\usepackage{pgfplots} 
\usepackage{pgfgantt}
\usepackage{pdflscape}
\pgfplotsset{compat=newest} 
\pgfplotsset{plot coordinates/math parser=false}
\usepackage{grffile}
\usetikzlibrary{plotmarks}
\usepackage{amsthm}

\usepackage{amsmath}%
\usepackage{MnSymbol}%
\usepackage{wasysym}
\usepackage{subcaption}
\usepackage{soul}

\usepackage[mathscr]{euscript}
\usepackage{cite}
\usepackage{amsfonts}
\usepackage{pgf}
\usepackage{mathtools}
\usepackage{graphicx}
\usepackage[ruled,vlined]{algorithm2e}
\usepackage{array}
\usepackage{tikz}
\usepackage[utf8]{inputenc}
\usepackage{pgfplots} 
\usepackage{pgfgantt}
\usepackage{pdflscape}
\pgfplotsset{compat=newest} 
\pgfplotsset{plot coordinates/math parser=false}
\usepackage{grffile}
\usetikzlibrary{plotmarks}
\usepackage{tikzscale}
\usepackage{caption}
\usepackage{subcaption}
\usepackage{array}
\usepackage{tikz}
\usepackage[utf8]{inputenc}
\usepackage{pgfplots} 
\usepackage{pgfgantt}
\usepackage{pdflscape}
\pgfplotsset{compat=newest} 
\pgfplotsset{plot coordinates/math parser=false}
\usepackage{grffile}
\usetikzlibrary{plotmarks}
\usepackage[T1]{fontenc}
\usepackage{graphicx}
\usepackage{caption}
\usepackage{overpic}
\usepackage{array}
\usepackage{color}
\usepackage{url}
\usepackage{tikz}
\usepackage[utf8]{inputenc}
\usepackage{pgfplots} 
\usepackage{cite}

\newtheorem{theorem}{Theorem}

\def\CN{\mathcal{C}\mathcal{N}} 

\usepackage[usestackEOL]{stackengine} 

\begin{document}

	\title{On the Average Secrecy Performance of Satellite Networks in Short Packet Communication Systems}

	\author[$\dagger$]{Ramin Hashemi}
        \author[$\ddagger$]{Graciela Corral Briones}
        \author[$\dagger$]{Risto Wichman}
        
        \affil[$\dagger$]{Department of Information and Communications Engineering, Aalto University, Espoo, Finland,\authorcr Emails: ramin.hashemi, risto.wichman@aalto.fi}
        \affil[$\ddagger$]{National University of Córdoba, Córdoba, Argentina,\authorcr Email: graciela.corral@unc.edu.ar}

	\maketitle

	\begin{abstract}
        This paper investigates the secrecy performance of satellite networks in short packet communication systems under shadowed Rician fading (SRF). We derive a lower bound on the average achievable secrecy rate in the finite blocklength regime (FBL) and provide analytical insights into the impact of key secrecy-related performance indicators (KPIs). Monte Carlo simulations validate the theoretical framework, and demonstrate that increasing the blocklength and improving the legitimate receiver's signal-to-noise ratio (SNR) enhance secrecy, while a stronger eavesdropper degrades it. Additionally, we show that directional antenna patterns can effectively reduce information leakage and provide secure satellite communications in the short packet regime. These findings offer valuable guidance for designing secure and efficient satellite-based communication systems, particularly in IoT and space-based networks.
	\end{abstract}
	

	\IEEEpeerreviewmaketitle
    \section{Introduction}
    \bstctlcite{IEEEexample:BSTcontrol}
    Short packet communication plays a crucial role in the Internet of Things (IoT) and real-time status updates for low Earth orbit (LEO) satellites, which operate at altitudes between 400 and 2000 km above the Earth \cite{Pingyue, Kodheli}. For example, certain remote sensors, such as those deployed in oceans, transmit telemetry data— including updated coordinates and status information—via small/short packets to reduce latency \cite{Shunkai}. A practical application of this is Iridium’s Short Burst Data (SBD) satellite service, which supports global IoT connectivity by transmitting messages of up to approximately 340 bytes \cite{saari2014small}. Moreover, studies indicate that the quality of service (QoS) of short packet transmissions over satellite links is significantly impacted in various mission-critical applications \cite{Kokkoniemi25}. This requires us to guarantee the reliability and secrecy of the satellite channels as much as possible, particularly in emergency scenarios \cite{Thien}.

    Ensuring secure communication in satellite networks is of paramount importance, especially when ground stations are in a mission-critical task such that having reliable and secret communication is essential \cite{ZhangSecureSatellite,Shunkai}. A fundamental study in \cite{YangWireTap} has proposed new achievability and converse bounds that refine existing secrecy capacity bounds and provide the tightest known results for the second-order coding rate in discrete memoryless and Gaussian wiretap channels in finite blocklength (FBL) regime. These principal derivations provide significant insights for secure communication design in modern IoT and satellite networks as well \cite{Tatar,Bhargav,QingqingAoI}. The upper/lower bound introduced in \cite{YangWireTap} is applicable to wiretap channels that incorporate additive white Gaussian noise.

    Various papers have studied secret communication in satellite networks \cite{ZhangSecureSatellite,QingqingAoI,Tatar,Bhargav,Talgat}. A comprehensive study on outage probability, and physical layer security in uplink was done in \cite{ZhangSecureSatellite} by considering analytical secrecy rate expressions in infinite blocklength regime. The authors in \cite{Tatar} analyzed the security of FBL transmissions in wiretap fading channels using a new secrecy metric called average information leakage. Here, only Rayleigh and Rician fading channels were considered in analytical derivations. The study in \cite{Bhargav} presents new analytical expressions for the probability of strictly positive secrecy capacity and a lower bound on the secure outage probability over $\kappa - \mu$ fading channels. In \cite{QingqingAoI}, the authors have proposed strategies to reduce the Age of Information (AoI) in satellite-based IoT networks using Rate-Splitting Multiple Access (RSMA) for short-packet transmissions under shadowed-Rician fading conditions. Then, closed-form mathematical models for block error rate (BLER) and average AoI (AAoI) were proposed to support efficient power allocation via deep reinforcement learning. The authors in \cite{Talgat} have studied secure IoT-based LEO satellite networks by proposing availability, successful, and secure communication probabilities. To the best of our knowledge, there is no study on the average secrecy rate performance of satellite networks in FBL regime. 

    In this paper, we will shed some light on the average performance of short packet satellite networks in a secret communication scenario. 
    The performance analysis of short packet communications in terrestrial networks does not directly apply to satellite networks because of the differences in propagation environments. 
    There are various satellite channel models in the literature, such as Loo's model \cite{ChunLoo}, shadowed Rician fading \cite{AbdiJournal}, and Gaussian mixture shadowing model \cite{modelinguplink}. In this paper, we employ the model in \cite{AbdiJournal} due to its simplicity and accuracy.

    \section{System Model}
    \label{SysModelSec}
    Consider an uplink short packet transmission from a ground base station in which the transmit antenna has a specific power pattern. The ground station is named Alice, and the legitimate receiving satellite is called Bob. As shown in Fig. \ref{fig:sysmodel}, there is an eavesdropper, i.e., Eve satellite, that tries to wiretap the communication channel between Alice and Bob. The complex channel coefficients from Alice to Bob or Eve are drawn from the shadowed Rician fading (SRF) model. The instantaneous complex channel between ground station, and satellite $i$, $i \in \{\text{B}, \text{E}\}$ at time $t$ is written as \cite{AbdiJournal}
    \begin{flalign}
        h_i(t) = r_1^i(t)\exp(ja_1^i(t)) + r_2^i(t)\exp(j\beta^i), \quad j^2 = -1 \label{complexSig}
    \end{flalign}
    where $r_1^i$ follows a Rayleigh distribution, $r_1^i \sim \text{Rayleigh}(b_i)$, with $2b_i=\mathbb{E}[(r_1^i)^2]$ representing the average scatter power ($\mathbb{E}[\cdot]$ denotes the statistical expectation operator). The second component is distributed as $r_2^i\sim\text{Nakagami}(m_i,\Omega_i)$ ($m_i\geq0$, $\Omega_i\geq0$) such that $\Omega_i$ is the average power of the LOS component. The phase component $a_1^i(t)$ has a uniform distribution between $[0,2\pi]$, and the $\beta^i$ is a deterministic LOS phase. The channel envelope $|h_i(t)|$ of the model in \eqref{complexSig} has a shadowed Rician distribution in which the probability density function (PDF) is given by \cite{AbdiJournal}
    \begin{flalign}
        f_{|h_i|}(x) &= \left( \frac{2 b_i m_i}{2 b_i m_i + \Omega} \right)^{m_i} \frac{x}{2 b_i} \exp \left( -\frac{x^2}{2 b_i} \right) \nonumber \\
        &\quad \times {}_1F_1 \left( m_i, 1, \frac{\Omega_i x^2}{2 b_i (2 b_i m_i + \Omega_i)} \right),\label{PDFh2}
    \end{flalign}
    where ${}_1F_1 \left(\cdot,\cdot,\cdot\right)$ is the confluent hypergeometric function \cite{gradshteyn2014table}. Practical and useful characteristics on the moment generating function (MGF), and $\omega^{\text{th}}$ order moments of \eqref{PDFh2} were already derived in \cite{AbdiJournal} which are
    \begin{flalign}
        \phi_i(\omega)  &= \mathbb{E}[|h_i|^\omega] = \left( \frac{2 b_i m_i}{2 b_i m_i + \Omega_i} \right)^{m_i} (2 b_i)^{\omega/2} \Gamma \left( \frac{\omega}{2} + 1 \right) \nonumber \\
        &\quad \times {}_2F_1 \left( \frac{\omega}{2} + 1, m_i, 1, \frac{\Omega_i}{2 b_i m_i + \Omega_i} \right), \label{moment}\\
        \text{M}_{|h_i|^2}(\eta) & = \mathbb{E}\left[e^{-|h_i|^2 \eta }\right] \nonumber \\ &= \frac{(2 b_i m_i)^{m_i} (1 + 2 b_i \eta)^{m_i-1}}{\left[ (2 b_i m_i + \Omega_i)(1 + 2 b_i \eta) - \Omega_i \right]^{m_i}}, \quad \eta \geq 0,  
    \end{flalign}
    where $\Gamma \left( \cdot \right)$ is the gamma function \cite{gradshteyn2014table}.

    \begin{figure}[t]
		\centering
		\includegraphics[trim = 0cm 0cm 0cm 0cm,scale=0.5]{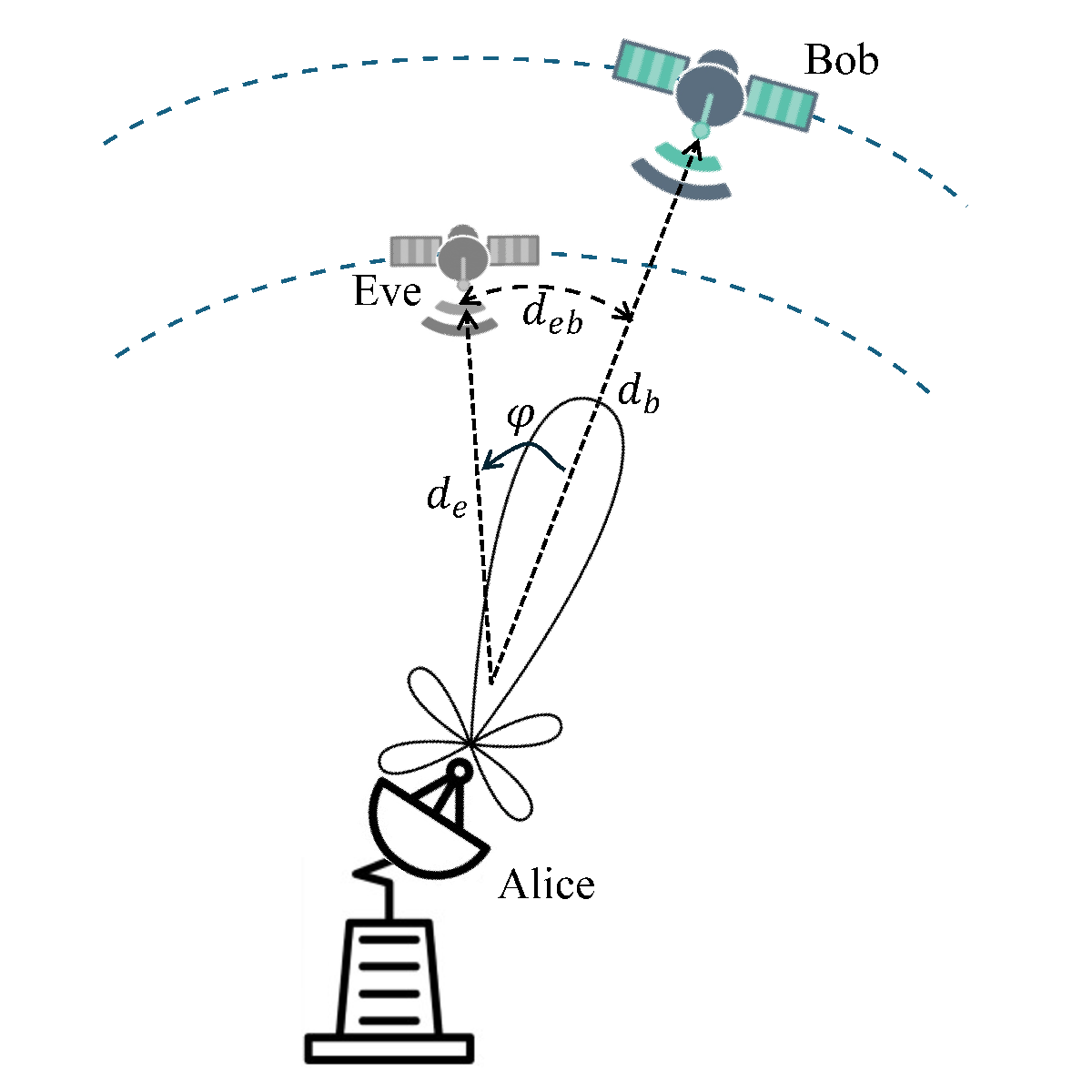}
		\caption{Illustration of the considered system model and the communication links between Alice, Bob, and Eve.}
		\label{fig:sysmodel}
    \end{figure}
    \setlength{\textfloatsep}{10pt}

Our study does not specifically target any satellite system in which ground stations and satellites have their characteristic antenna beam patterns. Instead, we characterize satellite receivers by their signal-to-noise ratios (SNRs), without considering their antenna beams.
For the radiation pattern of the ground station, we follow the ITU recommendation for interference assessment \cite{ITURS4656}. Thus, the reference radiation pattern toward Eve w.r.t.  the main beam toward Bob in Fig. \ref{fig:sysmodel} is defined by
\begin{flalign}
 G &=   32- 25 \log \varphi\;\; \text{dBi for}\;\; \varphi_\mathrm{min} \leq\varphi < 48^\circ \\
 &= -10\;\; \text{dBi for}\;\; 48^\circ \leq \varphi\leq180^\circ \nonumber
\end{flalign}
where $\varphi_\mathrm{min} = 1^\circ$.


    Finally, given the Bob, and Eve channel gains and statistics $|h_i|$ $i\in\{B,E\}$, the signal-to-noise ratio (SNR) in Bob (legitimate) and eavesdropper (Eve) are expressed as
    \begin{flalign}
        \text{SNR}_\text{B} = & P'_{\text{A}} \times G'_{\text{A}} \times G_{\text{B}} \times \text{L}_{\text{B}}^{\text{free space}} \times |h_\text{B}|^2/\sigma^2, \\ \nonumber 
        \text{SNR}_{\text{E}} = & P'_{\text{A}}\times G'_{\text{A}} \times G_{\text{E}} \times \text{L}_{\text{E}}^{\text{free space}} \times |h_\text{E}|^2/\sigma^2,
    \end{flalign}
    where $G'_{\text{A}}$, and $G'_{\text{B}}$ are the antenna gains of the Alice and Bob respectively. The transmit power of the Alice is defined as $P'_{\text{A}}$ and $\sigma^2$ is the additive white Gaussian noise power distributed as $\CN(0,\sigma^2)$. $\text{L}_{i}^{\text{free space}}=\frac{\lambda}{4\pi d_i}$ is the free space path loss where $\lambda$ is the wavelength, and $d_i$ represents the distance of satellite $i\in{\text{B},\text{E}}$ from Alice, as shown in Fig. \ref{fig:sysmodel}. Let us denote the relative arc length distance between Eve and Bob satellites as $d_{eb}$. Considering a 2D plane, we have $d_{eb}=\varphi[\text{radian}]\times d_{e}$. For notation simplicity, let us denote $P_{\text{A}}=\frac{P'_{\text{A}} G'_{\text{A}}}{\sigma^2}$.
    
    In wiretap channels, the instantaneous achievable secrecy rate $R_s$ in FBL regime can be expressed as \cite{YangWireTap}
    \begin{flalign}
    R_s \approx C_s- \sqrt{\frac{V_\text{B}}{n}} Q^{-1}(\epsilon_\text{B}) - \sqrt{\frac{V_\text{E}}{n}} Q^{-1}(\delta),
    \label{FBLs}
    \end{flalign}
    where 
    \begin{flalign}
        C_s = \left[ C_\text{B} - C_\text{E} \right]^+, \label{C_s}
    \end{flalign}
    is the secrecy capacity, and $[x]^+ = \max(0, x)$ denotes the positive part function. In addition, $C_\text{B}=\log_2(1+\text{SNR}_\text{B})$ and $C_\text{E}=\log_2(1+\text{SNR}_\text{E})$ denote the capacities of the Bob (legitimate) and eavesdropper (Eve) channels, respectively. $\epsilon_\text{B}$ is the target reliability probability for Bob, $\delta$ is the secrecy constraint (or information leakage to Eve), and $n$ is the channel block length. $V_\text{B}$ and $V_\text{E}$ are the channel dispersions of the Bob and eavesdropper channels, respectively. For $i\in\{\text{B},\text{E}\}$ the dispersion expressions are defined as
    \begin{flalign}
    V_i = \frac{1}{(\ln{2})^2} \cdot \frac{\text{SNR}_i^2 + 2 \cdot \text{SNR}_i}{\left(1 + \text{SNR}_i\right)^2},
    \end{flalign}
    where $\ln(\cdot)$ is the natural logarithm. The total channel dispersion is defined as
    \begin{flalign}
    V_C = V_\text{B} + V_\text{E} - \frac{1}{(\ln{2})^2} \cdot \frac{\text{SNR}_\text{E}^2 + 2 \cdot \text{SNR}_\text{E}}{\left(1 + \text{SNR}_\text{E}\right) \cdot \left(1 + \text{SNR}_\text{B}\right)},
    \end{flalign}
    $Q^{-1}(\cdot)$ denotes the inverse of the standard Q-function $Q(x) = \frac{1}{\sqrt{2\pi}}\int_{x}^{\infty}e^{-\nu^2/2}d\nu$. 
    Note that when $n \rightarrow \infty$ we have secrecy channel capacity in infinite blocklength regime, i.e., $R_s \rightarrow C_s$.

    \vspace{-2mm}
    \section{Average Achievable Secrecy Rate}
    \label{LB}
    Our aim is to derive a tractable analytical framework for the secrecy performance of the proposed system model. It is known that the expected average performance of the system achievable rate provides an insightful perspective regarding the bahviour of the Eve satellites. Therefore, the secrecy FBL rate in \eqref{FBLs} can be used to calculate the average achievable secrecy rate (AASR) as 
    \begin{flalign}
        \mathbb{E}\left[R_s\right] \approx \overset{\overline{R}_s}{\overbrace{\mathbb{E}\left[C_s- \sqrt{\frac{V_\text{B}}{n}} Q^{-1}(\epsilon_\text{B}) - \sqrt{\frac{V_\text{E}}{n}} Q^{-1}(\delta)\right]}}, \label{E_FBL}
    \end{flalign}
    where the expectation is taken over the channel distribution, which follows shadowed Rician fading. It is worth noting that direct calculation of \eqref{E_FBL} involves complex integrals thanks to confluent hypergeometric functions, and a closed-form solution is highly unlikely possible to the best of our knowledge. Therefore, we resort to deriving a lower bound for  \eqref{E_FBL} as it is useful also for many optimization frameworks that focus on maximizing the system key performance indicators (KPIs). In other words, maximizing the lower bound achievable rate makes more sense than working with upperbound expressions. 
    \begin{theorem}
    \label{theo1}
    A tractable lower bound for the average achievable secrecy rate in \eqref{E_FBL} is given by
    \begin{flalign}
        \mathbb{E}\left[R_s\right] \geq \tilde{R},
    \end{flalign}
   where $\tilde{R}$ is expressed as 
    \begin{flalign}
        \tilde{R}  = &\log_2\left(1+P_\text{A} G_\text{B} L_\text{B}^{\text{free space}}\exp\left(2\phi'_\text{B}(0)\right)\right) \\ \nonumber 
        & - \log_2\left(1+\overline{\text{SNR}}_\text{E}\right)  \\ \nonumber 
        & -\frac{Q^{-1}(\epsilon_\text{B})}{\ln{2}\sqrt{n}} \cdot \sqrt{1-\frac{1}{\left(1 + \overline{\text{SNR}}_\text{B}\right)^2}}  \\ \nonumber 
        & -\frac{Q^{-1}(\delta)}{\ln{2}\sqrt{n}} \cdot \sqrt{1-\frac{1}{\left(1 + \overline{\text{SNR}}_\text{E}\right)^2}}
    \end{flalign}
    with
    \begin{flalign}
        \overline{\text{SNR}}_\text{B} & \triangleq \mathbb{E}\left[\text{SNR}_\text{B}\right]=P_\text{A} G_\text{B} L_\text{B}^{\text{free space}}\phi_\text{B}(2), \\ \nonumber 
        \overline{\text{SNR}}_\text{E} & \triangleq \mathbb{E}\left[\text{SNR}_\text{B}\right]=P_\text{A} G_\text{E} L_\text{E}^{\text{free space}}\phi_\text{E}(2),\\ \nonumber 
        \phi_\text{i}(2) & = \mathbb{E}[|h_i|^2] = 2 b_i\left( \frac{2 b_i m_i}{2 b_i m_i + \Omega_i} \right)^{m_i}  \nonumber \\
         & \quad\quad\quad\ \ \ \times {}_2F_1 \left( 2, m_i, 1, \frac{\Omega_i}{2 b_i + \Omega_i} \right), \ \ i \in \{\text{B},\text{E}\} 
    \end{flalign}
    and $\phi'_{\text{B}}(0)$ is given in \eqref{phip_B}.
    
    \end{theorem}
    \begin{proof}
    Let us consider the first two capacity expressions $C_\text{B}$, and $C_\text{E}$ 
    in \eqref{C_s} (note that $C_\text{B}\geq C_\text{E}$ otherwise the secrecy rate is invalid). By noting the concavity of the function $\ln(1+x)$ and using Jensen inequality as \cite{boyd2004convex}
    \begin{flalign}
    \mathbb{E}[ \ln(1+x)] \leq \ln(1+\mathbb{E}[x]), 
    \end{flalign}
    therefore, $-\mathbb{E}[C_\text{E}]=-\mathbb{E}[\log_2(1+\text{SNR}_\text{E})]\geq -\log_2(1+\mathbb{E}[\text{SNR}_\text{E}])$. In addition, we use a variable transform as $\text{SNR}_\text{B}=\exp(u)$, $u\in\mathbb{R}$ for Bob's capacity term $C_\text{B}$, and noting that the function $\log_2(1+\exp(u))$ is a convex function and according to Jensen inequality, there is a lower bound as given by
    \begin{flalign}
        \mathbb{E}[C_\text{B}] \geq \log_2\left(1+\exp\left(\mathbb{E}[u]\right)\right),
    \end{flalign}
    where $\mathbb{E}[u]=\mathbb{E}[\ln(\text{SNR}_\text{B})]$. Here, to compute the expected value, we leverage the advantage of moment functions in \eqref{moment} 
    \begin{flalign}
        \mathbb{E}[\ln(\text{SNR}_\text{B})] = \ln\left(P_{\text{A}} G_\text{B} L_\text{B}^{\text{free space}}\right) \mathbb{E}\left[\ln(|h_\text{B}|^2)\right],
    \end{flalign}
    where we know that
    \begin{flalign}
      \mathbb{E}\left[\ln(|h_\text{B}|^2)\right] = \frac{\partial \phi_\text{B} (2\omega)}{\partial \omega}\Big|_{\omega=0}  = 2\phi'_{\text{B}}(0),
    \end{flalign}
    by calculating the partial derivative and evaluating at $\omega=0$
    \begin{flalign}
        \phi'_{\text{B}} =\phi'_{\text{B}}(\omega=0) & =  -\frac{\gamma}{2} + \frac{\ln(2)}{2} + \frac{\ln(b_\text{B})}{2}+\frac{1}{2}\left(\frac{2b_\text{B}m_{\text{B}}}{2b_\text{B}m_\text{B}+\Omega_\text{B}}\right)^{m_\text{B}} \nonumber \\ 
         &  \times {}_2F'_1 \left(1,m_\text{B},1,\frac{\Omega_\text{B}}{2 b_\text{B} (2 b_\text{B} m_\text{B} + \Omega_\text{B})}\right),
         \label{phip_B}
    \end{flalign}
   where $\gamma\approx0.5772$ is the Euler-Mascheroni constant, and the partial derivative in ${}_2F'_1 \left(a,b,c,x\right)$ is taken with respect to the first argument, i.e., $a$. Please note that the derivation result is computed in a single point which is tractable, and computationally feasible.

   Now, we focus on the dispersion terms in \eqref{E_FBL}. We note that the dispersion-related term $\tilde{V}_i=\sqrt{V_i}$ is a concave function in terms of $\text{SNR}_i\geq0$\footnote{Note that $Q^{-1}(x)\geq0$ for $0\leq x\leq0.5$. Typically $x\ll0.5$, e.g., $x=10^{-4}$.}. This can be easily proved by calculating the second derivative and showing that $\tilde{V}_i''\leq0$. Therefore, because of the negative sign in \eqref{E_FBL} $-\tilde{V}_i=-\sqrt{V_i}$ is a convex function, and we can apply Jensen inequality \cite{boyd2004convex}
   \begin{flalign}
    \mathbb{E}\left[-\sqrt{V_i}\right] \geq -\frac{1}{\ln{2}} \cdot \sqrt{\left(1-\frac{1}{\left(1 + \mathbb{E}\left[\text{SNR}_i\right]\right)^2}\right)},
    \end{flalign}
    Therefore, by substituting all obtained lower-bound expressions into \eqref{E_FBL}, the proof is completed.
    \end{proof}
    
    In addition to the lower bound derived in Theorem \ref{theo1}, there is also another method to derive an approximation for $\mathbb{E}[C_\text{B}]$ which is based on Taylor series expansion of $\ln(1+x)$ around $x_0=\mathbb{E}[x]$ and taking another $\mathbb{E}[\cdot]$ from both sides \cite{RaminTVT}
   \begin{flalign}
        \mathbb{E}[C_\text{B}] \approx \log_2\left(1+\mathbb{E}[\text{SNR}_\text{B}]\right) - \frac{\mathbb{V}[\text{SNR}_\text{B}]}{2\ln(2)(1+\mathbb{E}[\text{SNR}_\text{B}])}, 
    \end{flalign}
   with $\mathbb{V}[\cdot]$ is the variance operation. By exploiting lower bound expressions for the other convex terms in \eqref{E_FBL}, we propose another approximation for the lower bound of the average secrecy rate
   \begin{flalign}
        \tilde{R}\approx \log_2\left(1+\overline{\text{SNR}}_{\text{B}}\right) - \frac{\mathbb{E}[\text{SNR}_{\text{B}}^2]-\left(\overline{\text{SNR}}_\text{B}\right)^2}{2\ln(2)(1+\overline{\text{SNR}}_{\text{B}})},\label{lb2}
    \end{flalign}
    where similar to previous derivations $\mathbb{E}[\text{SNR}_{\text{B}}^2]$ is given by
    \begin{flalign}
        \mathbb{E}[\text{SNR}_{\text{B}}^2] = (P_\text{A} G_\text{B} L_\text{B}^{\text{free space}})^2\phi_\text{B}(4).
    \end{flalign}
    and $\phi_\text{B}(4)=\mathbb{E}[|h_\text{B}|^4]$ can be easily computed via replacing corresponding parameters in \eqref{moment}
    \begin{flalign}
        \mathbb{E}[|h_\text{B}|^4]  &= 8( b_\text{B})^{2}\left( \frac{2 b_\text{B} m_\text{B}}{2 b_\text{B} m_\text{B} + \Omega_\text{B}} \right)^{m_\text{B}}  \nonumber \\
        &\quad \times {}_2F_1 \left( 3, m_\text{B}, 1, \frac{\Omega_\text{B}}{2 b_\text{B} m_\text{B} + \Omega_\text{B}} \right).
    \end{flalign}

    \section{Numerical Results}
    In this section, we present Monte Carlo simulation results to evaluate the proposed lower bounds' accuracy. Additionally, we analyze key system-level performance metrics, including the secrecy rate and the extent of information leakage that the Eve satellite can infer. Table \ref{tab1} summarizes the parameters used in the simulations. In the presented figures, the lower bound refers to the evaluation of the analytical expression derived in Theorem \ref{theo1}, and approx. refers to \eqref{lb2}.

	\begin{table}[t]
		\caption{Simulation parameters.}
		\centering
		\begin{tabular}{ l  m{4cm} }
			\hline
			Parameter & Default value \\ \hline
			Bob's SNR ($\text{SNR}_{\text{B}}$)  & 5 dB  \\ 
                Eve's SNR ($\text{SNR}_{\text{E}}$) & -3 dB \\
                Channel blocklength ($n$) & 500 bits \\
                Information bits (packet size) & 300 bits \\
                Information leakage ($\delta$) & $10^{-3}$ \\
                Target reliability ($\epsilon$) & $10^{-3}$ \\
                Bob satellite antenna gain & $\text{G}_{\text{B}}=1$ \\
                Eve satellite antenna gain & $\text{G}_{\text{E}}=1$ \\
                $\{\Omega_i$, $m_i$, $b_i\}$, $i \in \{\text{B, E}\}$ & $\{0.515, 26, 0.005\}$ \\
                Satellite distance $d_{i}$ $i \in \{\text{B},\text{E}\}$ & 2000 km \\
                Minimum phase in gain profile ($\varphi_{\text{min}}$) & 1$^\circ$ \\ 
			\hline
		\end{tabular}
		\label{tab1}
	\end{table}

    Fig. \ref{fig:Rate_vs_blocklengths} illustrates the average secrecy rate as a function of the channel blocklength in three cases, namely as in the derived analytical lower bounds, Monte Carlo simulations, and the asymptotic secrecy capacity. From the negligible gap between analytical lower bounds and the exact secrecy rate calculations, the effectiveness of the proposed analytical framework is validated. The results confirm that as blocklength increases, the secrecy rate approaches its asymptotic limit, and the lower bounds provide a tight approximation to the exact secrecy rate.
    
    \begin{figure}[t]
        \centering
        \includegraphics[width=1\linewidth]{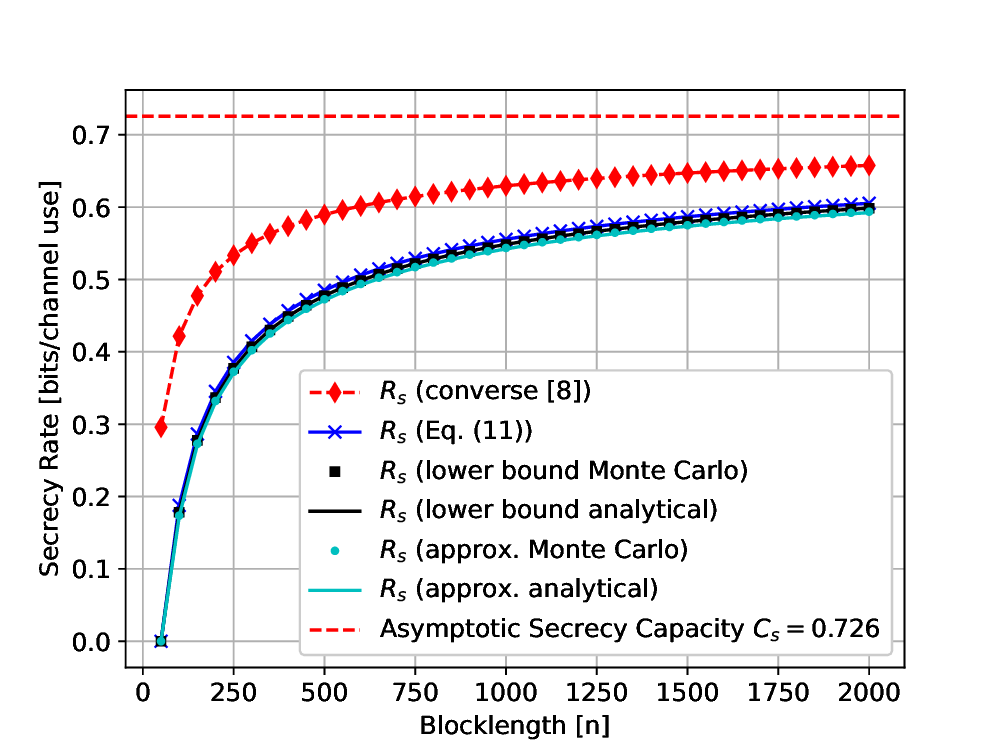}
        \caption{Average secrecy rate vs. channel blocklength.}
        \label{fig:Rate_vs_blocklengths}
    \end{figure}

    In Fig. \ref{fig:Rate_vs_SNR_B_dB_values} the secrecy rate as a function of Bob’s SNR is shown. The secrecy rate increases with higher Bob's SNR, indicating that improved legitimate channel conditions enhance secure communication. The derived analytical lower bounds closely match the Monte Carlo results, validating the theoretical derivations.
    \begin{figure}[t]
        \centering
        \includegraphics[width=1\linewidth]{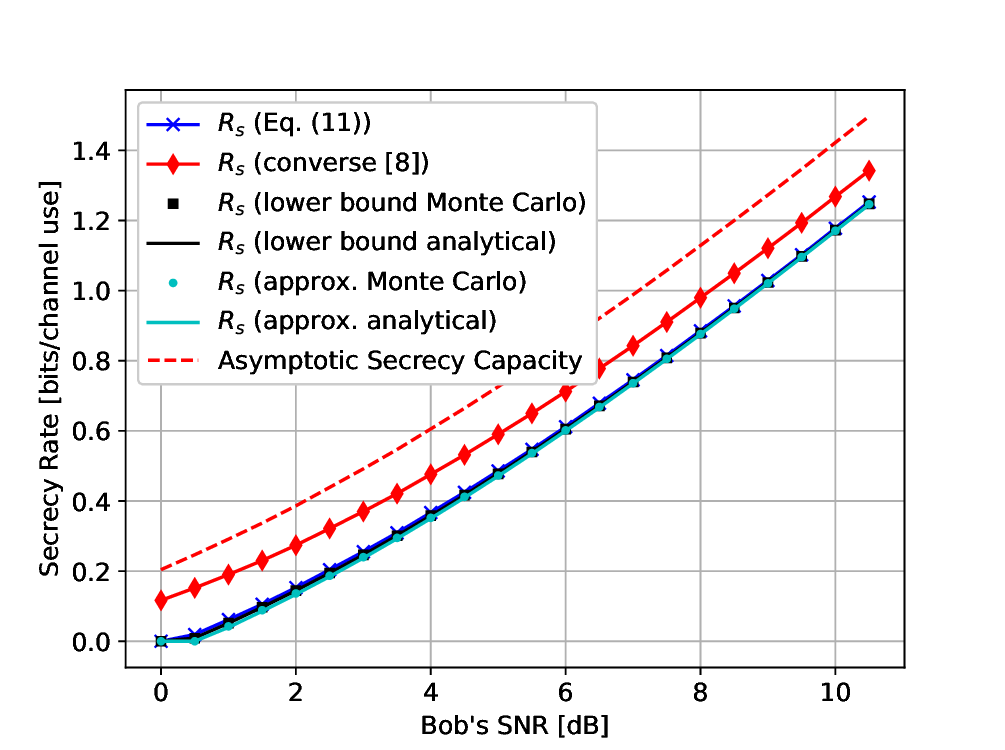}
        \caption{Average secrecy rate vs. Bob's SNR.}
        \label{fig:Rate_vs_SNR_B_dB_values}
    \end{figure}

    The performance of the average achievable rate in terms of Eve's SNR is plotted in Fig. \ref{fig:Rate_vs_SNR_E_dB_values}. As Eve's SNR increases, the secrecy rate decreases, demonstrating the vulnerability of the system to better eavesdropper channel conditions. The derived lower bounds provide accurate predictions of the secrecy performance, which is also supported by Monte Carlo simulation results.
    \begin{figure}[t]
        \centering
        \includegraphics[width=1\linewidth]{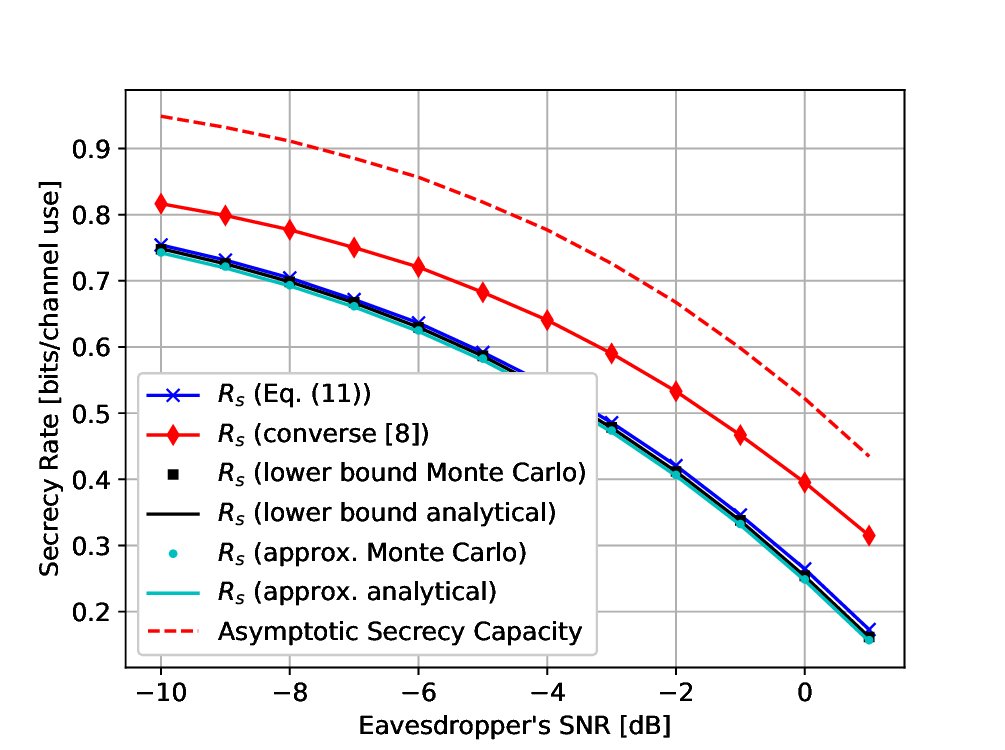}
        \caption{Average secrecy rate vs. Eve’s SNR.}
        \label{fig:Rate_vs_SNR_E_dB_values}
    \end{figure}

    Fig. \ref{fig:theta_vs_informationLeakage_values} depicts the information leakage as a function of the angle $\varphi$ (transmitter antenna orientation w.r.t. Eve), and the inter-satellite distance $d_{eb}$. The results show that information leakage increases for smaller $\varphi$, which implies that directional antenna patterns can effectively enhance secrecy. From another perspective, when the Eve satellite is at an arc distance of ~45 km from Bob (\(d_{eb} = 45\) km), the information leakage is approximately \(\delta \approx 10^{-4}\). This shows that the arc distance of the Eve satellite plays a crucial role in information leakage, and 
    the information leakage experienced by the eavesdropping satellite (Eve) is significantly influenced by the beam width and the proportional angle of the transmitting ground station. Notably, a change of a few degrees in the antenna orientation can lead to an information leakage increase of up to tenfold. 
    \begin{figure}[t]
        \centering
        \includegraphics[width=1\linewidth]{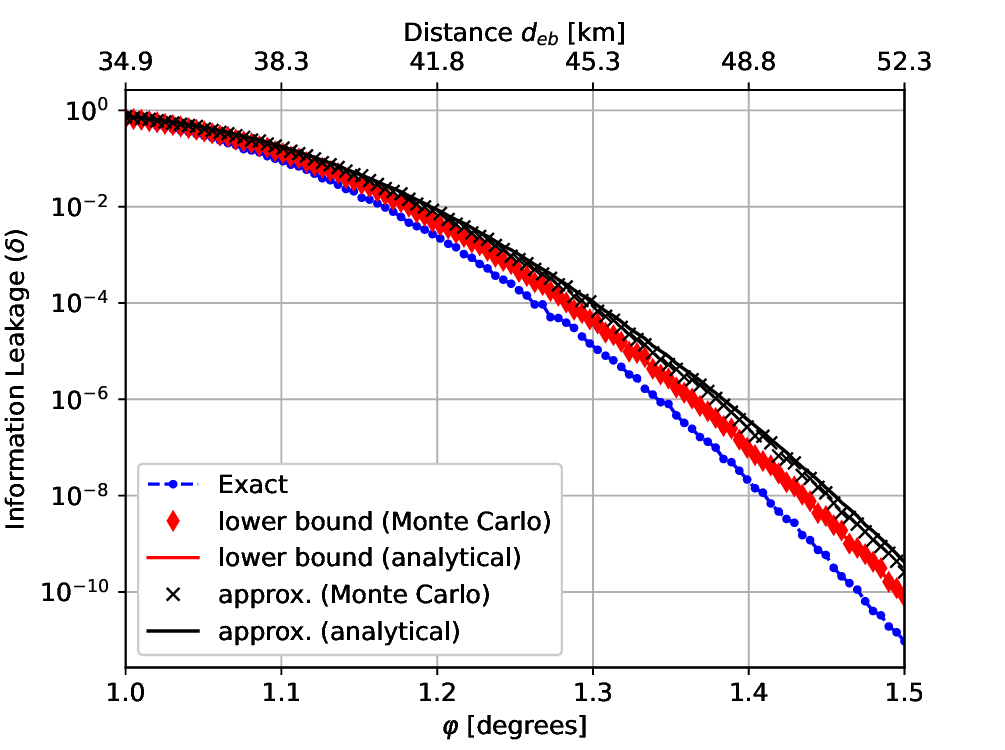}
        \caption{Impact of antenna orientation angle $\varphi$, and $d_{eb}$ on information leakage. A higher $\varphi$ results in reduced leakage.}
        \label{fig:theta_vs_informationLeakage_values}
    \end{figure}

\vspace{-2mm}
\section{Conclusion}
In this paper, we analyzed the average secrecy performance of satellite networks in short packet communication systems in FBL regime over shadowed Rician fading channels. We derived a lower bound for the achievable secrecy rate in the FBL regime and validated the theoretical framework through Monte Carlo simulations. The numerical results demonstrated that increasing blocklength and improving the legitimate user's SNR enhance secrecy performance, while higher eavesdropper SNR degrades it. Furthermore, directional antenna patterns reduce information leakage and enhance overall security. These findings offer insights into designing secure satellite communication systems for reliable and secret applications such as certain IoT and space-based networks.

	\bibliographystyle{IEEEtran}
    \bibliography{refs}
	
\end{document}